\documentclass[final]{siamart0516}
\pdfoutput=1
\usepackage[ascii]{inputenc}
\usepackage{hyperref,bbm,url,braket,microtype,mathtools,amsfonts,mathrsfs,cleveref}
\usepackage[T1]{fontenc}
\usepackage[USenglish]{babel}
\usepackage[basic,roman]{complexity}
\allowdisplaybreaks[4]

\newsiamthm{thm}{Theorem}\crefname{thm}{Theorem}{Theorems}
\newsiamthm{lem}{Lemma}\crefname{lem}{Lemma}{Lemmas}
\newsiamthm{cor}{Corollary}\crefname{cor}{Corollary}{Corollaries}
\newsiamthm{prb}{Problem}\crefname{prb}{Problem}{Problems}
\newsiamthm{dfn}{Definition}\crefname{dfn}{Definition}{Definitions}
\crefname{section}{section}{sections}
\crefname{subsection}{section}{sections}

\DeclareMathOperator{\GL}{GL}
\DeclareMathOperator{\SU}{SU}
\DeclareMathOperator{\Kron}{Kron}
\DeclareMathOperator{\Sym}{Sym}
\DeclareMathOperator{\diag}{diag}
\DeclareMathOperator{\tr}{tr}

\newcommand{\QQ}{\mathbb Q}
\newcommand{\CC}{\mathbb C}
\renewcommand{\PP}{\mathbb P}
\newcommand{\ZZ}{\mathbb Z}
\newcommand{\RR}{\mathbb R}

\newcommand{\TheTitle}{Membership in moment polytopes is in NP and coNP}
\newcommand{\TheAuthors}{P.~B\"urgisser, M.~Christandl. K.D.~Mulmuley, M.~Walter}
\headers{\TheTitle}{\TheAuthors}
\ifpdf
\hypersetup{
  pdftitle={\TheTitle},
  pdfauthor={\TheAuthors}
}
\fi

\begin{document}
\title{Membership in moment polytopes is in NP and coNP%
\thanks{Submitted to the editors on Nov 18, 2015.
\funding{PB was partially supported by DFG grant BU 1371/3-2.
MC was supported by a Sapere Aude grant of the Danish Council for Independent Research, an ERC Starting Grant, an SNSF Professorship, and the VILLUM FONDEN via the QMATH Centre of Excellence (grant no.\ 10059).
KM was supported by NSF grant CCF-1017760.
MW was supported by the Simons Foundation, FQXI, and AFOSR (grant no.\ FA9550-16-1-0082).}}}

\author{Peter B\"urgisser%
\thanks{Technische Universit\"at Berlin, Institut f\"ur Mathematik, MA 3-2, Stra\ss{}e des 17.\ Juni 136, 10623 Berlin, Germany (\email{pbuerg@math.tu-berlin.de}).}
\and
Matthias Christandl%
\thanks{
QMATH, Department of Mathematical Sciences, University of Copenhagen, Universitetsparken 5, 2100 Copenhagen, Denmark (\email{christandl@math.ku.dk}).}
\and
Ketan D.~Mulmuley%
\thanks{The University of Chicago, Computer Science Dept., 1100 East 58th Street, Chicago, IL 60637 (\email{mulmuley@cs.uchicago.edu}).}
\and
Michael Walter%
\thanks{Stanford University, Department of Physics, 382 Via Pueblo Mall, Stanford, CA 94305 (\email{michael.walter@stanford.edu}).}
}

\maketitle

\begin{abstract}
  We show that the problem of deciding membership in the moment polytope associated with a finite-dimensional unitary representation of a compact, connected Lie group is in $\NP\cap\coNP$.
  This is the first non-trivial result on the computational complexity of this problem, which naively amounts to a quadratically-constrained program.
  Our result applies in particular to the Kronecker polytopes, and therefore to the problem of deciding positivity of the stretched Kronecker coefficients.
  In contrast, it has recently been shown that deciding positivity of a single Kronecker coefficient is $\NP$-hard in general~\cite{ikenmeyer2015vanishing}.
  We discuss the consequences of our work in the context of complexity theory and the quantum marginal problem.
\end{abstract}



\section{Introduction and summary of results}

Moment polytopes are convex polytopes that describe invariants of Hamiltonian manifolds.
Their study has a long and rich history in mathematics~\cite{schur23klasse,horn1954doubly,kostant1973convexity,atiyah1981convexity,ness1984stratification,kirwan1984convexity} and in physics~\cite{guillemin1990symplectic}, most recently in quantum information theory in the context of the quantum marginal problem~\cite{christandl2006spectra,daftuar2005quantum,klyachko2004quantum,christandl2007nonzero,altunbulak2008pauli,walter2013entanglement,christandl2014eigenvalue,walter2014thesis}.
Moment polytopes and their underlying representation-theoretic data have also become of interest in computer science, since they are possible sources of representation-theoretic obstructions that may lead to new complexity-theoretic lower bounds~\cite{mulmuley2001geometric,burgisser2011geometric,burgisser2011overview}.

In this paper, we consider the computational complexity of deciding membership in a given moment polytope.
We show that, for a broad class of groups and representations, the problem of deciding membership in the associated moment polytope is in the complexity classes $\NP$ and $\coNP$.
Before presenting our general results, we discuss the important special case of the moment polytopes associated with the Kronecker coefficients, which are of particular interest in several of the applications mentioned above.
Recall that a \emph{Young diagram} with $k$ boxes is an integer partition $\lambda_1 + \dots + \lambda_m = k$ where $\lambda_1 \geq \dots \geq \lambda_m > 0$.
We write $\lambda \vdash k$ for a Young diagram with $k$ boxes.
Young diagrams can be visualized as arrangement of blocks with $\lambda_i$ blocks in the $i$-th row; thus $m$ is called the height of the diagram.
Any Young diagram with height no larger than $m$ can be understood as a highest weight of the unitary group $U(m)$ and, by restriction, of $\SU(m)$.
Now consider the irreducible representation of $G = \SU(m)^3$ on $M = (\CC^m)^{\otimes 3}$ by tensor products, $(U_A,U_B,U_C) \mapsto U_A \otimes U_B \otimes U_C$.
Then any triple of Young diagrams $\lambda = (\lambda_A,\lambda_B,\lambda_C)$ can be understood as a highest weight of $G$.
The multiplicity $g(\lambda_A,\lambda_B,\lambda_C)$ of the corresponding irreducible representation of $G$ in the $k$-th symmetric power $\Sym^k(M)$ is known as the \emph{Kronecker coefficient}, where we may assume that all three diagrams have the same number of boxes $k$. 
The corresponding moment polytopes are known as the \emph{Kronecker polytopes}, parameterized by $m$, and they are given by
\[
  \Kron(m) = \overline{\left\{ \frac {(\lambda_A,\lambda_B,\lambda_C)} k : \lambda_A, \lambda_B, \lambda_C \vdash k, g(\lambda_A,\lambda_B,\lambda_C) > 0 \right\}} \subseteq \RR^{3m}.
\]
We remark that $\Kron(m)$ is known to be a convex polytope of dimension $3(m-1)$~\cite{vergnewalter2014inequalities}.

We consider the problem \textsc{KronPolytope} of deciding whether $(\lambda_A,\lambda_B,\lambda_C)/k \in \Kron(m)$, given as input a triple of Young diagrams $\lambda_A, \lambda_B, \lambda_C \vdash k$ (each specified by its row lengths encoded in binary), where $m$ denotes the maximum among the heights of the three Young diagrams.
Equivalently, \textsc{KronPolytope} is the problem of deciding whether there exists some positive integer $l > 0$ such that the stretched Kronecker coefficient $g(l \lambda_A, l \lambda_B, l \lambda_C) > 0$.

Our main result in the case of the Kronecker polytopes then is the following theorem:

\begin{thm}
\label{thm:kronecker}
  The problem \textsc{KronPolytope} is in $\NP\cap\coNP$.
\end{thm}

That is, there exists polynomially-sized certificates such that both membership and non-membership can be verified in polynomial time.
We discuss the implications of \cref{thm:kronecker} in \cref{subsec:discussion of results} below.

\smallskip

We now consider the general case.
Let $G$ denote a compact, connected Lie group and $M$ a unitary representation of $G$.
For each integer $k > 0$, we denote by $m_{G,M,k}(\lambda)$ the multiplicity of the irreducible $G$-representation with highest weight $\lambda$ in $\Sym^k(M)$, the $k$-th symmetric power of~$M$.
Then the pairs $(k,\lambda)$ for which $m_{G,M,k}(\lambda) > 0$ form a finitely generated semigroup and so the following set is a rational convex polytope, called the \emph{moment polytope} associated with the $G$-representation $M$:
\[
  \Delta_G(M) = \overline{\left\{ \frac \lambda k : m_{G,M,k}(\lambda) > 0 \right\}} \subseteq \Lambda^*_G \otimes_\ZZ \RR \subseteq i\mathfrak t^*
\]
Here, $\Lambda^*_G$ denotes the weight lattice and $\mathfrak t^*$ the dual of the Lie algebra of a maximal torus $T$ of $G$.
We shall assume that the moment polytope is of maximal dimension, i.e., that $\dim \Delta_G(M) = \dim T$, the rank of the Lie group (equivalently, that generic points in the projective space $\PP(M)$ have discrete $G$-stabilizer).

We are interested in the problem \textsc{MomentPolytope} of deciding whether $\lambda/k \in \Delta_G(M)$, given as input a compact, connected Lie group $G$, a finite-dimensional unitary representation $M$ with moment polytope of maximal dimension, a highest weight $\lambda$, and a positive integer $k$. Equivalently, \textsc{MomentPolytope} is the problem of deciding whether there exists some positive integer $l > 0$ such that the stretched multiplicity $m_{G,M,lk}(l\lambda) > 0$.
We discuss the precise encoding of the input in \cref{subsec:specification} below.
Roughly speaking, the group is specified in terms of Dynkin diagrams and the representation in terms of its highest weights, given by its coefficients in binary with respect to a basis of fundamental weights, together with the total dimension of the representation $M$ in unary. The latter is a natural requirement, as it allows the algorithms to run in polynomial time in the dimension of the representation, which can be exponential in the specification of the highest weights alone. In the case of the Kronecker coefficients, this requirement is vacuous, as the dimension is only of polynomial size in the specification of $\lambda_A,\lambda_B,\lambda_C$, and so it is not hard to see that the problem \textsc{MomentPolytope} is indeed a proper generalization of \textsc{KronPolytope}.

The main result of this paper is the following theorem, which generalizes \cref{thm:kronecker} above.

\begin{thm}
\label{thm:general}
  The problem \textsc{MomentPolytope} is in $\NP\cap\coNP$.
\end{thm}

Sets of defining inequalities for $\Delta_G(M)$ have been computed in~\cite{berenstein2000coadjoint,klyachko2004quantum,ressayre2010geometric,vergnewalter2014inequalities}.
To prove \cref{thm:kronecker,thm:general}, our principal ingredient is the recent description from~\cite{vergnewalter2014inequalities}.
We also rely on an alternative, geometric characterization of $\Delta_G(M)$, often taken as its definition, whose equivalence has been established by Mumford~\cite{ness1984stratification}.
In both cases, a major challenge is to show that these mathematical results can be made effective, i.e., that approximations can be found that give rise to polynomial-sized certificates, and that these certificates can in turn be verified efficiently.

\subsection{Discussion of results}
\label{subsec:discussion of results}

We note that \cref{thm:general,thm:kronecker} are non-trivial complexity-theoretic results.
On the one hand, from the representation-theoretic point of view, all known upper bounds on the stretching factor $l$ required to witness membership of some $\lambda/k$ in the moment polytope can be exponential in the input specification (even when restricted to invariants, see, e.g., \cite{derksen2001polynomial}).
On the other hand, the geometric description of the moment polytopes naively amounts to a quadratically constrained program, which are NP-hard in general.
In contrast, \cref{thm:general,thm:kronecker} strongly suggest that \textsc{MomentPolytope} and \textsc{KronPolytope} are \emph{not} NP-hard problems, for otherwise $\NP = \coNP$, which is widely regarded as implausible (e.g., \cite{goldreich2008computational}).
Thus the situation is similar to that of the integer factorization problem (likewise in $\NP\cap\coNP$), the unknotting problem (in $\NP\cap\coNP$, assuming the generalized Riemann hypothesis), and the graph isomorphism problem (in $\NP\cap\coAM$).
This is in remarkable contrast to the recent result in~\cite{ikenmeyer2015vanishing} that deciding positivity of a single Kronecker coefficient is NP-hard in general.

It might be conjectured that membership in moment polytopes can in fact be decided in polynomial time.
This is known to be true for the Horn polytopes, which geometrically characterize the eigenvalues of triples of Hermitian matrices that add up to zero, $A + B + C = 0$.
However, the proof in this case relies precisely on the fact that the \emph{Littlewood-Richardson coefficients}, which are the associated representation-theoretic coefficients and in fact special Kronecker coefficients, are saturated and that their positivity can be decided in polynomial time~\cite{knutson2001honeycombs,blasiak2011gct3,burgisser2013deciding}.
Saturation does not hold in general, and in particular not for the Kronecker coefficients.
Moreover, as we have just discussed, deciding positivity is in general an NP-hard problem, so this strategy of proof cannot be generalized. In contrast, our strategy of proof circumvents this barrier and may be seen as a first step towards establishing the conjecture.

Moment polytopes also play a fundamental role in quantum physics.
Let $\psi_{ABC} \in \CC^m \otimes \CC^m \otimes \CC^m$ be a unit vector and $\psi_A, \psi_B, \psi_C$ the corresponding reduced density matrices (defined precisely in~\cref{eq:rdm} below).
Then the Kronecker polytope $\Kron(m)$ can be identified with the set of possible triples $(r_A(\psi), r_B(\psi), r_C(\psi))$, where we write $r_A(\psi)$ for the spectrum of $\psi_A$, etc.
Thus \textsc{KronPolytope} corresponds precisely to the \emph{one-body quantum marginal problem} for three quantum particles~\cite{christandl2006spectra,daftuar2005quantum,klyachko2004quantum,christandl2007nonzero,walter2014thesis,vergnewalter2014inequalities}.
We explain this connection more carefully in \cref{subsec:kron np} as part of our proof that \textsc{KronPolytope} is in \NP.
We may consider more general quantum marginal problems where the set of observables is given by the Lie algebra of a complex reductive group.
This includes the one-body quantum marginal problem for several distinguishable particles, for indistinguishable particles with bosonic or fermionic statistics, and for spin and orbit degrees of freedom~\cite{altunbulak2008pauli,walter2014thesis,vergnewalter2014inequalities}.
Any such problem can be phrased in terms of moment polytopes of representations and therefore corresponds to special cases of the general membership problem, \textsc{MomentPolytope} (cf.\ \cref{subsec:general np}).
Our results show that the computational complexity is in $\NP \cap \coNP$ given the encoding of problem instances described in \cref{subsec:specification}.
We may also consider an alternative encoding where the dimension of the representation in unary is not part of the input.
As the dimension can be exponentially large in the bitsize of the remaining input, our \cref{thm:general} implies that the general membership problem is in $\NEXP \cap \coNEXP$.
In some cases of physical interest, it is known that the complexity is in $\QMA(2)$~\cite{liu2007quantum}, which is contained in $\NEXP$.

Finally, we remark that for a fixed group and representation (such as for Young diagrams with bounded height), the membership problem is trivial, since it concerns only a single polytope which can be precomputed.
Likewise, it is known that in this case positivity can be decided and indeed that the coefficients can be calculated precisely in polynomial time~\cite{christandl2012computing}.

\subsection{Organization of the paper}

In \cref{sec:kronecker}, we first prove our result in the important special case of the family of Kronecker polytopes (\cref{thm:kronecker}).
We will follow the proof strategy for the general membership problem, but our presentation will not rely on expert knowledge in the representation theory of Lie groups.
Then, in \cref{sec:general} we prove our general result, where the group and representation defining the moment polytope are part of the input (\cref{thm:general}).

\subsection{Notation and conventions}
\label{subsec:notation}

We write $\lambda \vdash k$ for a Young diagram with $k$ boxes and $\ell(\lambda)$ for the number of rows of a Young diagram.
We define $\ZZ^m_0 := \{ (x_i) \in \ZZ^m : \sum_{i=1}^m x_i = 0 \}$.
We always encode natural numbers in binary, rational numbers in terms of their numerator and denominator, numbers in $\QQ[i]$ by their real and imaginary part, and Young diagrams by listing their row lengths.
All logarithms are with respect to base 2.
We write $\# S$ for the cardinality of a finite set $S$.
Throughout this article, we will work with several norms:
For vectors $v$ in a real vector space, we denote by $\lVert v \rVert_2$ the Euclidean norm and by $\lVert v \rVert_\infty$ the maximum norm.
For vectors $\psi$ in complex Hilbert space, we denote by $\lVert \psi \rVert$ the norm induced by the inner product.
For linear operators $A$ acting on Hilbert space, we denote by $\lVert A \rVert$ the operator norm, by $\lVert A \rVert_1$ the trace norm, and by $\lVert A \rVert_F$ the Frobenius norm.
Finally, in \cref{sec:general} we introduce norms $\lVert-\rVert_{i \mathfrak g}$ and $\lVert-\rVert_{i \mathfrak g^*}$ from the Killing form of a Lie algebra $\mathfrak g$.

\section{The Kronecker polytopes}
\label{sec:kronecker}

In this section, we will prove our complexity result for the Kronecker polytopes (\cref{thm:kronecker}).
While this result can also be obtained as a consequence of our general result (\cref{thm:general}), which we prove in \cref{sec:general} below, the exposition in this section contains all essential ideas while not requiring expert knowledge in representation theory.
Our notation and terminology will match precisely the one used in \cref{sec:general} below.

\subsection{Inequalities for the Kronecker polytopes}
\label{subsec:kronecker ieqs}

Let $m > 0$ be a positive integer.
We define $P(m) := \{ (x_A,x_B,x_C) \in \RR^{3m}: \sum_{i=1}^m x_{A,i} = \sum_{j=1}^m x_{B,j} = \sum_{k=1}^m x_{C,k} = 1 \}$.
We note that $P(m)$ is an affine space of dimension $3(m-1)$.
We will call $\Phi(m) := \{ (e_i, e_j, e_k) : i,j,k=1,\dots,m \} \subseteq P(m)$ the set of \emph{weights}, where $e_i$ denotes the $i$-th standard basis vector of $\RR^m$, and $N(m) := \{ (e_i - e_j, 0, 0) : i > j \} \cup \{ (0, e_i - e_j, 0) : i > j \} \cup \{ (0, 0, e_i - e_j) : i > j \}$ the set of \emph{negative roots}.
For any $H = (H_A,H_B,H_C) \in (\ZZ^m_0)^3$ and $z \in \ZZ$, we define the following three subsets:
\begin{align*}
  \Phi(H=z) &= \{ \varphi \in \Phi(m) : \varphi \cdot H = z \}, \\
  \Phi(H<z) &= \{ \omega \in \Phi(m) : \omega \cdot H < z \}, \\
  N(H<0) &= \{ \alpha \in N(m) : \alpha \cdot H < 0 \}.
\end{align*}

\begin{dfn}[\cite{vergnewalter2014inequalities}]
\label{dfn:ressayre kron}
  A \emph{Ressayre element} is a pair $(H,z)$, where $H = (H_A,H_B,$ $H_C) \in (\ZZ^m_0)^3$ and $z \in \ZZ$, such that the following conditions are satisfied:
  \begin{enumerate}
    \item \emph{Admissibility:} The points in $\Phi(H=z)$ span an affine hyperplane in $P(m)$.
    \item \emph{Trace condition:} $\# N(H<0) = \# \Phi(H<z)$.
    \item \emph{Determinant condition:} Consider the following matrix $D_{H,z}$ whose rows are indexed by elements $\omega \in \Phi(H<z)$ and whose columns are indexed by elements $\alpha \in N(H<0)$,
    \begin{equation}
    \label{eq:kron det matrix}
      (D_{H,z})_{\omega,\alpha} = \begin{cases}
      X_\varphi & \text{if } \varphi := \omega - \alpha \in \Phi(H=z), \\
      0 & \text{otherwise},
    \end{cases}
    \end{equation}
    where the $X_\varphi$ are indeterminates.
    By the trace condition, $D_{H,z}$ is a square matrix, so that we can form the \emph{determinant polynomial} $d_{H,z} := \det D_{H,z}$, and the condition is that $d_{H,z}$ should be non-zero. 
  \end{enumerate}
\end{dfn}

We observe that the number of Ressayre elements is finite (up to overall rescaling).
We have the following description of the Kronecker polytopes in terms of finitely many inequalities~\cite{vergnewalter2014inequalities}:
\begin{equation}
\label{eq:kronecker ressayre}
  \Kron(m) = \{ r \in P_+(m) : r \cdot H \geq z \text{ for all Ressayre elements $(H,z)$} \},
\end{equation}
where $P_+(m) = \{ (r_A,r_B,r_C) \in P(m) \;:\; r_{X,1} \geq \dots \geq r_{X,m} \; (\forall X=A,B,C) \}$ is the positive Weyl chamber.
It is known that $\Kron(m) \subseteq P(m)$ is a convex polytope of maximal dimension $3(m-1)$~\cite{vergnewalter2014inequalities}.
Let us call a facet of $\Kron(m)$ \emph{non-trivial} if it is \emph{not} of the form $r_{X,i} \geq r_{X_i+1}$.
\Cref{eq:kronecker ressayre} implies that any non-trivial facet is necessarily given by a Ressayre element.

\subsection{\textsc{KronPolytope} is in coNP}
\label{subsec:kron conp}

A problem instance for \textsc{KronPolytope} is given by three Young diagrams $\lambda_A,\lambda_B,\lambda_C \vdash k$.
We now describe a polynomial-time algorithm that takes as input the problem instance $\lambda_A,\lambda_B,\lambda_C \vdash k$ together with a certificate that consists of a triple $(H,z,p)$, where $H \in (\ZZ^m_0)^3$, $z \in \ZZ$ and $p \in \ZZ^{\#\Phi(H=z)}$, where $m$ is the maximal number of rows in $\lambda_A$, $\lambda_B$ and $\lambda_C$.

The algorithm proceeds as follows:
We first check the conditions in \cref{dfn:ressayre kron} to verify that $(H,z)$ is a Ressayre element for $\Kron(m)$:
\begin{enumerate}
  \item Admissibility:
    The number of weights $\#\Phi(m)$ is $m^3$ and each weight lives in a space of dimension $3m$.
    For each weight $\varphi\in\Phi(m)$, we can check whether $\varphi\in\Phi(H=z)$ by verifying that the inner product with $H$ satisfies $H\cdot\varphi=z$.
    Thus we can in polynomial time determine $\Phi(H=z)$ and compute the rank of the polynomial-size matrix with columns $\begin{psmallmatrix}\varphi\\-1\end{psmallmatrix}$ for $\varphi\in\Phi(H=z)$.
    The element $(H,z)$ is admissible if and only if the rank is equal to $3(m-1)$.
  \item Trace condition:
    As there are $O(m^2)$ negative roots and $m^3$ weights, each of which lives in a space of dimension $3m$ and can be constructed efficiently, both cardinalities can be computed and compared in polynomial time.
  \item Determinant:
    We construct the matrix $D_{H,z}(p)$ defined as in \cref{eq:kron det matrix} for $X = p$.
    The matrix is of polynomial size and we can therefore compute its determinant $d_{H,z}(p)$ exactly in polynomial time.
    We accept if and only if $d_{H,z}(p) \neq 0$.
\end{enumerate}
At this point we are sure that $(H,z)$ defines a non-trivial facet of the Kronecker polytope (the trivial inequalities are automatically satisfied since $\lambda_A$, $\lambda_B$ and $\lambda_C$ are partitions).
In the last step of the algorithm, we verify that this facet indeed separates $(\lambda_A,\lambda_B,\lambda_C)/k$ from the polytope by checking that
\[ H \cdot (\lambda_A,\lambda_B,\lambda_C) < k z. \]
It is clear that the algorithm will accept only if $(\lambda_A,\lambda_B,\lambda_C)/k \not\in \Kron(m)$.

We will now show that, conversely, if $(\lambda_A,\lambda_B,\lambda_C)/k \not\in \Kron(m)$ then there always exists a polynomial-sized certificate $(H,z,p)$ such that the algorithm accepts.
For this, we need the following basic estimate:

\begin{lem}
\label{lem:kron siegel}
  Any non-trivial facet of $\Kron(m)$ can be described by a Ressayre element $(H,z)$ with
  \begin{equation}
  \label{eq:kron siegel}
    \max~\{ \lVert H \rVert_\infty, \lvert z \rvert \} \leq (4m)^{3m},
  \end{equation}
  where $\lVert H \rVert_\infty := \max_{i=1}^{3m}~\lvert H_i \rvert$.
\end{lem}
\begin{proof}
  The admissibility condition in \cref{dfn:ressayre kron} asserts that any Ressayre element $(H,z)$ is the normal vector of an affine hyperplane in $P(m)$ spanned by some affinely independent set of weights $\omega_1, \dots, \omega_{3(m-1)} \in \Phi(H=z)$.
  Therefore, $(H,z) \in \ZZ^{3m+1}$ is an integral solution to the following linear system of equations:
  \begin{align*}
    &H \cdot \omega_1 - z = 0, \;\dots,\; H \cdot \omega_{3(m-1)} - z = 0, \\
    &H_A \cdot (1,\dots,1) = 0, H_B \cdot (1,\dots,1) = 0, H_C \cdot (1,\dots,1) = 0
  \end{align*}
  Note that we have $M = 3m$ equations for $N = 3m + 1$ unknowns, and the absolute value of the coefficients is at most $B=1$.
  Therefore, Siegel's lemma~\cite{hindry2000diophantine} ensures that there exists an integral solution with
  \[ \max~\{ \lVert H \rVert_\infty, \lvert z \rvert \} \leq (NB)^{M/(N-M)} = (3m + 1)^{3m} \leq (4m)^{3m}, \]
  as claimed by the lemma.
\end{proof}

The upshot of \cref{lem:kron siegel} is the following:
If $(\lambda_A,\lambda_B,\lambda_C)/k \not\in \Kron(m)$ then there exists a non-trivial facet separating it from the Kronecker polytope.
\Cref{lem:kron siegel} tells us that any such facet can be encoded by some Ressayre element $(H,z)$ that can be specified using no more than $O(m^2 \log m)$ bits.
Indeed, $(H,z)$ consists of $3m+1$ coefficients, each of which requires $O(m \log m)$ bits.

At last, consider the determinant polynomial $d_{H,z}$, which is a nonzero multivariate polynomial of degree $\#\Phi(H<z) = \#N(H<0) \leq m^3$ in $\#\Phi(H=z) \leq m^3$ variables.
The Schwartz-Zippel lemma~\cite[Corollary 1]{schwartz1980fast} shows that the fraction of points $p \in \{0,\dots,m^3\}^{\#\Phi(H=z)}$ with $d_{H,z}(p)=0$ is at most $\#\Phi(H<z) / (m^3 + 1) < 1$.
It follows that there exists some $p \in \{0,\dots,m^3\}^{\#\Phi(H=z)}$ such that $d_{H,z}(p) \neq 0$.
Note that $p$ can be specified using no more than $O(m^3 \log m)$ bits.

As the input size is $\Omega(m)$, the data $(H,z,p)$ together consists of a polynomial-sized certificate that will be accepted by the algorithm.
We conclude that the problem \textsc{KronPolytope} is in $\coNP$.

We remark that Alon's combinatorial Nullstellensatz~\cite[Theorem 1.2]{alon1999combinatorial} gives a much stronger bound than the Schwartz-Zippel lemma, and it would be interesting to see if it can be exploited to find even smaller certificates.

\subsection{\textsc{KronPolytope} is in NP}
\label{subsec:kron np}

We now show that \textsc{KronPolytope} is in $\NP$.
For this, we will use the following \emph{geometric} description of the Kronecker polytopes, which can be deduced more generally from Mumford's theorem~\cite{ness1984stratification}.
For any non-zero vector $\psi \in A \otimes B \otimes C = \CC^m \otimes \CC^m \otimes \CC^m$, the \emph{reduced density matrix} $\psi_A$ is the Hermitian operator on $\CC^m$ defined by duality in the following way:
\begin{equation}
\label{eq:rdm}
  \forall X_A\!: \qquad \tr \psi_A X_A = \frac {\braket{\psi | X_A \otimes \mathbbm 1_B \otimes \mathbbm 1_C | \psi}} {\braket{\psi|\psi}},
\end{equation}
where $X_A$ ranges over all Hermitian operators on $\CC^m$ and where $\braket{-|-}$ denotes the inner product on $\CC^m \otimes \CC^m \otimes \CC^m$. We likewise define $\psi_B$ and $\psi_C$. We observe that all three reduced density matrices are positive semidefinite operators with unit trace.
Finally, write $r_A(\psi)$, $r_B(\psi)$, $r_C(\psi)$ for the vector of eigenvalues of $\psi_A$, $\psi_B$, $\psi_C$, arranged in non-increasing order.
Then we have the following alternative characterization of the Kronecker polytopes:
\begin{equation}
\label{eq:kronecker polytope geometric}
  \Kron(m) = \{ (r_A(\psi), r_B(\psi), r_C(\psi)) : 0 \neq \psi \in \CC^m \otimes \CC^m \otimes \CC^m \}.
\end{equation}

We now describe a polynomial-time algorithm that takes as input the problem instance $\lambda_A,\lambda_B,\lambda_C \vdash k$ as well as a certificate that consists of a non-zero vector $\tilde\psi \in \QQ[i]^m \otimes \QQ[i]^m \otimes \QQ[i]^m$, where $m$ is the maximal number of rows in $\lambda_A$, $\lambda_B$ and $\lambda_C$.
Here, the certificate $\tilde\psi$ is specified in terms of $2m^3$ rational numbers, namely the real and imaginary parts of its $m^3$ coordinates $(\psi_{a,b,c})_{a,b,c=1}^m$ with respect to the standard product basis (cf.~\cref{subsec:notation}).

The algorithm proceeds in two steps:
First, we compute the reduced density matrices $\tilde\psi_A$, $\tilde\psi_B$, $\tilde\psi_C$ of $\tilde\psi$.
This can be done in polynomial time, since the polynomially many entries of the reduced density matrix $\tilde\psi_A$ can be computed by
\[ \forall a,a'=1,\dots,m\!: \qquad (\tilde\psi_A)_{a,a'} = \frac {\sum_{b,c=1}^m \tilde\psi_{a,b,c} \overline{\tilde\psi_{a',b,c}}} {\sum_{a,b,c=1}^m \lvert\tilde\psi_{a,b,c}\rvert^2}, \]
and likewise for $\tilde\psi_B$ and $\tilde\psi_C$.
Second, we accept if and only if
\begin{equation}
\label{eq:kron accept}
  \lVert (\tilde\psi_A,\tilde\psi_B,\tilde\psi_C) - (\diag \lambda_A,\diag \lambda_B,\diag \lambda_C)/k \rVert_F \leq \frac 1 2 \frac 1 k \frac 1 {(4m)^{4m}}.
\end{equation}
Here, we write $\diag v$ for the diagonal matrix with diagonal entries $v$, and $\lVert-\rVert_F$ denotes the Frobenius norm.
It is immediate that \cref{eq:kron accept} can be verified in polynomial time.

To show that \textsc{KronPolytope} is in \NP, we have to consider two cases.
In the case where $(\lambda_A,\lambda_B,\lambda_C)/k \in \Kron(m)$, we will prove that there exists a polynomially-sized certificate $\tilde\psi$ accepted by the algorithm (we will find that $b = O(m \log m + \log k)$ bits of precision suffice).
In the other case, where $(\lambda_A,\lambda_B,\lambda_C)/k \not\in \Kron(m)$, we will show that no certificate $\tilde\psi$ is accepted by the algorithm.

We first consider the case where the point $(\lambda_A,\lambda_B,\lambda_C)/k \in \Kron(m)$.
According to \cref{eq:kronecker polytope geometric}, there exists a unit vector $\psi \in \CC^m \otimes \CC^m \otimes \CC^m$ such that the eigenvalues $(r_A(\psi),r_B(\psi),r_C(\psi))$ of its reduced density matrices are equal to $(\lambda_A,\lambda_B,\lambda_C)/k$.
By applying a tensor product of local unitaries, $\psi \leadsto (U_A \otimes U_B \otimes U_C) \psi$, we may in fact arrange for the reduced density matrices to be \emph{equal} to the diagonal matrices with entries $\lambda_A/k$, etc. That is, there exists a unit vector $\psi$ such that
\begin{equation}
\label{eq:kron preimage}
  (\psi_A,\psi_B,\psi_C) = (\diag \lambda_A, \diag \lambda_B, \diag \lambda_C)/k.
\end{equation}
Now let $\tilde\psi \in \QQ[i]^m \otimes \QQ[i]^m \otimes \QQ[i]^m$ be a truncation of $\psi$ to $b$ bits in the real and imaginary part of each of its $m^3$ components.
The following lemmas bound the resulting error on the reduced density matrices as a function of the precision $b$:

\begin{lem}
\label{lem:hilbert space vs trace norm}
  Let $\psi \in \CC^D$ be a unit vector and $\tilde\psi \in \QQ[i]^D$ the vector obtained by truncating $\psi$ to $b$ bits in the real and imaginary parts of each of its components. Then we have
  \[ \lVert P_{\tilde\psi} - P_\psi \rVert_1 \leq 5 D^{1/4} 2^{-b/2}, \]
  where $P_\phi$ denotes the orthogonal projection onto the one-dimensional subspace $\CC\phi$ and $\lVert - \rVert_1$ the trace norm.
\end{lem}
\begin{proof}
  We start with the identity~\cite[(9.101)]{nielsen2010quantum}
  \[ \lVert P_{\tilde\psi} - P_\psi \rVert_1 = 2 \sqrt{1 - \lvert \braket{\psi|\frac {\tilde\psi} {\lVert \tilde\psi \rVert}} \rvert^2}, \]
  where $\lVert - \rVert$ denotes the Hilbert space norm and we have used that $\psi$ is a unit vector.
  We now write $\tilde\psi / \lVert \tilde\psi \rVert = \psi + \delta$. Then:
  \[ \lvert \braket{\psi|\frac {\tilde\psi} {\lVert \tilde\psi \rVert}} \rvert^2
   = \lvert \braket{\psi|\psi + \delta} \rvert^2
   \geq (1 - \lvert \braket{\psi|\delta} \rvert)^2
   \geq (1 - \lVert\delta\rVert)^2
   \geq 1 - 2 \lVert\delta\rVert.
  \]
  We now compute
  \begin{align*}
    &\lVert\delta\rVert
    = \lVert \frac {\tilde\psi} {\lVert \tilde\psi \rVert} - \psi \rVert
    \leq \lVert \frac {\tilde\psi} {\lVert \tilde\psi \rVert} - \tilde\psi \rVert + \lVert \tilde\psi - \psi \rVert
    = \lvert 1-\lVert \tilde\psi \rVert \rvert + \lVert \tilde\psi - \psi \rVert \\
    = &\lvert \lVert \psi \rVert - \lVert \tilde\psi \rVert \rvert + \lVert \tilde\psi - \psi \rVert
    \leq 2 \lVert \tilde\psi - \psi \rVert.
  \end{align*}
  By combining all three statements we obtain that
  \[ \lVert P_{\tilde\psi} - P_\psi \rVert_1 \leq 4 \sqrt{\lVert \tilde\psi - \psi \rVert}. \]
  As $\psi$ is a unit vector, the magnitude of each of its $D$ components is no larger than one.
  Thus the truncation incurs an absolute error of at most $2^{-b}$ on the real and imaginary parts.
  We conclude that
  \[
      4 \sqrt{\lVert \tilde\psi - \psi \rVert}
    \leq 4 (2 D 2^{-2b})^{1/4}
    \leq 5 D^{1/4} 2^{-b/2}. 
  \]
\end{proof}

\begin{lem}
\label{lem:kron truncation}
  For all $X=A,B,C$, we have that
  \[ \lVert \tilde\psi_X - \psi_X \rVert_F = O(m 2^{-b/2}). \]
\end{lem}
\begin{proof}
  We have the following sequence of inequalities,
  \[ \lVert \tilde\psi_X - \psi_X \rVert_F \leq \lVert \tilde\psi_X - \psi_X \rVert_1 \leq \lVert P_{\tilde\psi} - P_\psi \rVert_1 = O(m 2^{-b/2}), \]
  where the first inequality bounds the Frobenius norm in terms of the trace norm,
  the second inequality asserts that the trace norm does not increase under the partial trace $P_\phi \mapsto \phi_X$~\cite[(9.100)]{nielsen2010quantum},
  and the last inequality is \cref{lem:hilbert space vs trace norm} applied to the Hilbert space $M = (\CC^m)^{\otimes 3}$ of dimension $D = m^3$.
\end{proof}


As a direct consequence of \cref{lem:kron truncation,eq:kron preimage}, the truncation to $b$ bits leads to an error of at most
\begin{align*}
    &\lVert (\tilde\psi_A,\tilde\psi_B,\tilde\psi_C) - (\diag \lambda_A,\diag \lambda_B,\diag \lambda_C)/k \rVert_F \\
  =\;&\big( \lVert \tilde\psi_A - \psi_A \rVert_F^2 + \lVert \tilde\psi_B - \psi_B \rVert_F^2 + \lVert \tilde\psi_C - \psi_C \rVert_F^2 \big)^{1/2}
  = O(m 2^{-b/2})
\end{align*}
Comparing with \cref{eq:kron accept}, we find that we only need to choose $b = O(m \log m + \log k)$ bits of precision to produce a certificate that the algorithm accepts.
This is polynomial in the size of the problem instance $(\lambda_A,\lambda_B,\lambda_C)$, which is $\Omega(m)$ and $\Omega(\log k)$.
Thus there exists a polynomially-sized certificate that our algorithm accepts.

\medskip

Conversely, let us assume that in fact $(\lambda_A,\lambda_B,\lambda_C)/k \not\in \Kron(m)$.
We will use the following lemma, which in colloquial terms asserts that the ``slope'' of any facet of $\Kron(m)$ is never too steep.
More precisely:

\begin{lem}
\label{lem:distance kron}
  Let $\lambda_A, \lambda_B, \lambda_C \vdash k$. Then, if $(\lambda_A, \lambda_B, \lambda_C)/k \not\in \Kron(m)$, it has Euclidean distance at least
  \[ \frac 1 k \frac 1 {(4m)^{4m}} \]
  to the Kronecker polytope.
\end{lem}
\begin{proof}
  Consider a non-trivial facet that separates $(\lambda_A, \lambda_B, \lambda_C)/k$ from $\Kron(m)$.
  According to \cref{lem:kron siegel}, any such facet can be described by some Ressayre element $(H,z)$ satisfying \cref{eq:kron siegel}.
  We can therefore lower-bound the distance to $\Kron(m)$ in the following way:
  \[
    \frac {\lvert (\lambda_A, \lambda_B, \lambda_C)/k \cdot H - z \rvert}{\lVert H \rVert_2}
    = \frac 1 k \frac {\lvert (\lambda_A,\lambda_C,\lambda_C) \cdot H - kz \rvert} {\lVert H \rVert_2}
    \geq \frac 1 k \frac 1 {\lVert H \rVert_2},
  \]
  where we have used that both $(H,z)$ and $(\lambda_A,\lambda_B,\lambda_C)$ have integer coefficients.
  We now use that $\lVert H \rVert_2 \leq \sqrt{3m} \lVert H \rVert_\infty \leq 3m \lVert H \rVert_\infty$ and the upper bound \cref{eq:kron siegel} to conclude that
  \[
    \frac 1 k \frac 1 {\lVert H \rVert_2}
    \geq \frac 1 k \frac 1 {3 m \lVert H \rVert_\infty}
    \geq \frac 1 k \frac 1 {3m (4m)^{3m}}
    \geq \frac 1 k \frac 1 {(4m)^{4m}}.
  \]
\end{proof}

It follows from \cref{lem:distance kron,eq:kronecker polytope geometric} that the distance of the eigenvalues $(r_A(\tilde\psi),r_B(\tilde\psi),$ $r_C(\tilde\psi))$ of the reduced density matrices of any vector $0 \neq \tilde\psi \in \CC^m \otimes \CC^m \otimes \CC^m$ to $(\lambda_A,\lambda_B,\lambda_C)/k$ is never smaller than $1/k \cdot 1 / (4m)^{4m}$.
Together with the Wielandt-Hoffmann theorem (e.g., \cite[III.6.15]{bhatia2013matrix}), this implies that
\begin{align*}
\frac1k \frac1{(4m)^{4m}}
&\leq \lVert (r_A(\tilde\psi),r_B(\tilde\psi),r_C(\tilde\psi)) - (\lambda_A,\lambda_B,\lambda_C)/k \rVert_2 \\
&\leq \lVert (\tilde\psi_A,\tilde\psi_B,\tilde\psi_C) - (\diag \lambda_A,\diag \lambda_B,\diag \lambda_C)/k \rVert_F.
\end{align*}
In view of \cref{eq:kron accept}, our algorithm will therefore never accept if $(\lambda_A,\lambda_B,\lambda_C)/k \not\in \Kron(m)$.

We conclude that the problem \textsc{KronPolytope} is in $\NP$.
\Cref{subsec:kron conp,subsec:kron np} together establish \cref{thm:kronecker}.

\section{The general membership problem}
\label{sec:general}

We now turn to the membership problem for the moment polytope associated with an arbitrary finite-dimensional unitary representation $M$ of a compact, connected Lie group $G$.

We start with some classical facts from the theory of Lie groups and Lie algebras (see, e.g., \cite{kirillov2008introduction,hall2015lie}).
Recall that the Lie algebra of $G$ is of the form $\mathfrak g = \mathfrak g_1 \oplus \dots \oplus \mathfrak g_n$, where each $\mathfrak g_j$ is either the compact real form of a complex simple Lie algebra or one-dimensional abelian.
The Lie algebra $\mathfrak g$ does not determine $G$ completely.
However, any compact, connected $G$ has a finite covering group $\tilde G$, which is a product $G_1 \times \dots \times G_n$ where each factor is either the unique (up to isomorphism) connected, simply-connected Lie group $G_j$ with simple Lie algebra $\mathfrak g_j$ or a one-dimensional torus $U(1)$ with Lie algebra $i \RR$.
Both $G$ and $\tilde G$ have the same Lie algebra and we can reconstruct $\tilde G$ from $\mathfrak g$.
Furthermore, any representation $M$ of $G$ can be extended to a representation of $\tilde G$, and it is irreducible for $G$ if and only if it is irreducible for $\tilde G$.
In particular, both groups $G$ and $\tilde G$ lead to the same moment polytope: $\Delta_G(M) = \Delta_{\tilde G}(M)$.
On the one hand, this explains why it suffices for our purposes to specify the Lie algebra only, as we will do in \cref{subsec:specification} below.
On the other hand, this also allows us to assume without loss of generality that $G = \tilde G = G_1 \times \dots \times G_n$ in our analysis, and we shall do so throughout the remainder of this section.
In particular, this allows us to make the following choices in a coherent way:

For each Dynkin diagram, we fix a corresponding compact, connected, simply-connected Lie group~$G$ and simple Lie algebra~$\mathfrak g$, and choose once and for all a maximal torus $T \subseteq G$ and positive roots $P(G)$.
Let $\mathfrak g_\CC = \mathfrak g \otimes \CC$ denote the complexification.
For each root $\alpha$, we choose a basis vector $e_\alpha \in i \mathfrak g$ in the corresponding root space of $\mathfrak g_\CC$ as well as the corresponding coroot $h_\alpha \in i \mathfrak t$ such that the commutation relations of $\mathfrak{sl}_2(\CC)$, i.e., $[e_\alpha,e_{-\alpha}] = h_\alpha$ etc., are satisfied.
The coroots $h_\alpha$ span a lattice $\Lambda_G \subseteq i \mathfrak t$ of maximal rank, with basis given by the simple coroots (i.e., those corresponding to simple roots).
The elements of the dual basis of the simple coroots are known as the fundamental weights, they form a basis of the weight lattice $\Lambda^*_G \subseteq i \mathfrak t^*$.
Together, the $e_\alpha$ for arbitrary roots $\alpha$ and the simple coroots $h_\alpha$ determine a basis of $i \mathfrak g$.
Finally, we equip $i \mathfrak g$ with the Killing form, which is positive definite on $i \mathfrak g$, normalized such that the long roots have norm one.

For a one-dimensional torus $G = U(1)$, we have $G = T$ and $i\mathfrak g = i\mathfrak t = \RR$. We identify both the weight lattice $\Lambda^*_G$ and its dual lattice $\Lambda_G$ with $\ZZ$, choose $1 \in \ZZ$ as a basis vector of both lattices, and use the standard inner product on $\RR$.

Given an arbitrary group $G = G_1 \times \dots \times G_n$, the above choices determine a maximal torus $T$,
Lie algebra $\mathfrak t$,
Weyl group $W$,
positive roots $P(G)$ and negative roots $N(G) = -P(G)$,
the lattice $\Lambda_G \subseteq i \mathfrak t$ with basis the simple coroots $h_\alpha$ together with the basis vectors of the tori,
the dual weight lattice $\Lambda^*_G \subseteq i \mathfrak t^*$ with basis the fundamental weights $\omega_i$ together with the basis vectors of the tori,
a positive Weyl chamber $i \mathfrak t^*_+ = \{ r \in i \mathfrak t^* : r(h_\alpha) \geq 0 \; \forall \alpha \in P(G) \}$,
and a $G$-invariant inner product on~$i \mathfrak g$; we will denote the induced norm by $\lVert-\rVert_{i \mathfrak g}$.
By duality, we likewise obtain an inner product and norm $\lVert-\rVert_{i \mathfrak g^*}$ on $i \mathfrak g^*$ and on $i \mathfrak t^*$.
We note that the basis vectors of $\Lambda_G$ and $\Lambda_G^*$ have norm $\Theta(1)$ by our conventions (however we caution that they are not orthogonal).
At last, we obtain a basis of~$i \mathfrak g$ by adjoining to the basis of $\Lambda_G$ the basis vectors $e_\alpha$ of the root spaces; this also determines a dual basis of $i \mathfrak g^*$.
We will make repeated use of these objects in the following.

\subsection{Specification of the problem instance}
\label{subsec:specification}

Recall that \textsc{MomentPolytope} is the problem of deciding whether a given point $\lambda/k$ is an element of some moment polytope $\Delta_G(M)$.
A problem instance of \textsc{MomentPolytope} is thus given abstractly by a quadruple $(G,M,\lambda,k)$ consisting of a group $G$, a representation $M$, a highest weight $\lambda$, and an integer $k$.
We will now describe explicitly the specification in which we assume that this data is given to an algorithm.
For this, we follow~\cite{mulmuley2007geometric}; in particular, we will write $\langle X \rangle$ for the bitsize of an object $X$.
Thus the input size of a problem instance $(G,M,\lambda,k)$ is $\Theta(\langle G \rangle + \langle M \rangle + \langle \lambda \rangle + \langle k \rangle)$.

To specify the group $G$, we recall from the discussion above that its Lie algebra is of the form $\mathfrak g = \mathfrak g_1 \oplus \dots \oplus \mathfrak g_n$, where each $\mathfrak g_i$ is either one-dimensional abelian or the compact real form of a complex simple Lie algebra.
We will therefore specify $G$ in terms of its Lie algebra by listing the summands in such a decomposition:
For each $\mathfrak g_i$, we first record in a single bit whether it is abelian or not; in the latter case, we also specify the Dynkin diagram by giving its type ($A$--$D$, or one of the five exceptional families) and rank (in unary).
Thus $\langle \mathfrak g \rangle = \Omega(R)$, where $R$ is the rank of $\mathfrak g$ (i.e., the dimension of a maximal torus of $G$).

To specify the representation $M$, we note that $G$ is reductive so that $M$ can be decomposed into irreducible representations, $M = M_1 \oplus \dots \oplus M_p$.
Each irreducible representation $M_i$ is fully determined by its highest weight $\lambda_i$.
We will therefore specify $M$ by listing the highest weights $\lambda_1, \dots, \lambda_p \in \Lambda^*_G$, each specified in terms of its coefficients in binary with respect to the basis fixed above.
Finally, we require that the dimension of $M$ encoded \emph{in unary} is part of the specification of $M$, so that $\langle M \rangle = \Omega(\dim M)$.
This is a natural assumption that allows our algorithms to take polynomial time in the dimension of the representation (note that there exist irreducible representations whose dimensions are superpolynomial in the bitsize of the specification in terms of highest weight alone, e.g., an antisymmetric representation $\bigwedge^l \mathbb C^d$ of $\SU(d)$ for $d = 2l$).
On the other hand, we stress that the assumption does \emph{not} trivialize the problem.
This can in fact already be seen in the case of the Kronecker polytopes, where the dimension of $M = (\mathbb C^m)^{\otimes 3}$ is only polynomial in $\langle G \rangle$ and therefore can be generated (in unary) in polynomial time in $\langle G \rangle = \Theta(m)$.

Finally, to specify the highest weight $\lambda$ we likewise list its coefficients with respect to the basis fixed above (in binary), and the integer $k > 0$ is also specified in binary.

\subsection{Monomial bases and representation matrices}

To generalize our algorithms in \cref{sec:kronecker} to the general case, it will be necessary to perform various Lie-theoretic computations, such as determining the multiset of weights $\Phi(M)$ as well as computing representation matrices of the Lie algebra representation on $M$.
In this section we will explain how this can be done in polynomial time.
More precisely, we will establish the following results, which may be of independent interest:

\begin{lem}
\label{lem:multiset of weights in polynomial time}
  Given $G$ and $M$ as specified in \cref{subsec:specification}, the multiset of weights $\Phi(M)$ can be computed in polynomial time (as integer vectors with respect to the basis fixed at the beginning of \cref{sec:general}).
\end{lem}

\begin{lem}
\label{lem:representation matrices in polynomial time}
  Given $G$ and $M$ as specified in \cref{subsec:specification}, there exists a basis of weight vectors, indexed by $\Phi(M)$, such that the representation matrix of any of the basis vectors of $i \mathfrak g$ fixed at the beginning of \cref{sec:general} are rational and can be computed in polynomial time.
\end{lem}

It is plain that the set of negative roots $N(G)$ can also be computed in polynomial time.

\medskip

For the classical Lie groups, \cref{lem:multiset of weights in polynomial time,lem:representation matrices in polynomial time} can be established using well-known properties of Gelfand-Tsetlin or Molev patterns~\cite{gelfand1,gelfand2,molev1,molev2}. 
We will give a different proof, based on Lakshmibai's notion of a monomial basis of an irreducible representation~\cite{lakshmibai}, which
can be understood as a generalization of the Gelfand-Tsetlin basis to general semisimple complex Lie algebras.
This allows for a uniform proof of \cref{lem:multiset of weights in polynomial time} for all types, including the exceptional Lie groups.
Moreover, our proof of \cref{lem:representation matrices in polynomial time} for the exceptional Lie groups relies crucially on using monomial bases in order to reduce to type $A_n$.

In the following discussion, we shall assume that $\mathfrak g_\CC$ is simple and that $M$ is an irreducible representation (we will see below that this is without loss of generality).
Let $PS(G) \subseteq P(G)$ denote the set of simple roots.
For any $\alpha \in PS(G)$, let us denote by $\mathfrak{sl}_2^\alpha \subseteq \mathfrak g_\CC$ the corresponding copy of $\mathfrak{sl}_2(\CC)$, spanned by the generators $e_\alpha$, $f_\alpha := e_{-\alpha}$, and $h_\alpha$ that we had fixed at the beginning of \cref{sec:general}, and by $s_\alpha \in W$ the corresponding simple reflection.
Let $w \in W$ denote the longest Weyl group element, $\ell = \ell(w)$ its length, and fix a reduced decomposition into simple reflections, $w = s_{\alpha_\ell} \dots s_{\alpha_1}$.
Set $w_0 := 1$ and $w_r := s_{\alpha_r} \dots s_{\alpha_1}$ for any $r=1,\dots,\ell$, so that $w_\ell = w$.

Now suppose that $M = M_\lambda$ is an irreducible representation of highest weight $\lambda$.
We will denote the representation of the Lie algebra $\mathfrak g_\CC$ as well as its extension to the universal enveloping algebra by~$\pi$.
The highest weight vector $v_\lambda$ is an eigenvector of~$\mathfrak b$, the Borel subalgebra of $\mathfrak g_\CC$.
For any Weyl group element $w' \in W$, let $M_{\lambda,w'} \subseteq M_\lambda$ denote the corresponding \emph{Demazure module}, i.e., the $\mathfrak b$-module generated by $w' \cdot v_\lambda$. Note that $M_{\lambda,1} = \CC v_\lambda$, while $M_{\lambda,w_\ell} = M_\lambda$.
We will now describe Lakshmibai's inductive construction of sets $\mathcal M_{\lambda,w_r}$ of monomials of the form $x = f_{\alpha_s}^{n_s} \cdots f_{\alpha_1}^{n_1}$, with $\alpha_i \in PS(G)$ and each $n_i > 0$, such that each $\mathcal B_{\lambda,w_r} := \pi(\mathcal M_{\lambda,w_r}) v_\lambda$ is a basis of the Demazure module $M_{\lambda,w_r}$.
It will be useful to define $w(x) := s_{\alpha_s} \cdots s_{\alpha_1}$ and $\beta(x) := \lambda - \sum_{i=1}^s n_i \alpha_i$ for any such monomial.
For $r = 0$, we define $\mathcal M_{\lambda,w_0} := \{1\}$. Thus $\mathcal B_{\lambda,w_r} = \{ v_\lambda \}$ in this case.
For $r > 0$, we consider
\begin{equation}
\label{eq:monomials for induction}
  \mathcal N_{\lambda,w_{r-1}} := \{ x \in \mathcal M_{\lambda,w_{r-1}} : w_{r-1} \not\geq s_{\alpha_r} w(x)
 \},
\end{equation}
where $\geq$ denotes the Bruhat order, noting that $\mathcal M_{\lambda,w_{r-1}}$ has already been defined by the induction hypothesis.
For each $x \in \mathcal N_{\lambda,w_{r-1}}$, let $t_x$ denote the weight of $\pi(x) v_\lambda$ with respect to $\mathfrak{sl}_2^{\alpha_r}$, which is a positive integer~\cite{lakshmibai}.
That is, $t_x = \beta(x)(h_{\alpha_r})$, since $\beta(x)$ is the $\mathfrak g_\CC$-weight of $\pi(x) v_\lambda$.
We finally define
\begin{equation}
\label{eq:monomial basis}
  \mathcal M_{\lambda,w_r} :=
  \mathcal M_{\lambda,w_{r-1}}
  \cup
  \{ f_{\alpha_r}^i x : x \in \mathcal N_{\lambda,w_{r-1}}, i = 1, \dots, t_x \},
\end{equation}
which is a disjoint union.
Lakshmibai has shown that, for each $r=1,\dots,\ell$, $\mathcal B_{\lambda,w_r} = \pi(\mathcal M_{\lambda,w_r}) v_\lambda$ is a basis of the Demazure module $M_{\lambda,w_r}$ (see~\cite[Theorem 4.1]{lakshmibai} and its proof).
In particular, $\mathcal B_\lambda := \mathcal B_{\lambda,w_\ell}$ is a basis of the irreducible representation $M = M_\lambda$, indexed by the monomials in $\mathcal M_\lambda := \mathcal M_{\lambda,w_\ell}$. We call $\mathcal B_\lambda$ the \emph{Lakshmibai monomial basis} of $M_\lambda$.
From the perspective of computational complexity, the crucial observation is that $\mathcal B_\lambda$ can be constructed in polynomial time:

\begin{lem}
\label{lem:lakshmibai}
  Given $G$ and $M = M_\lambda$ as specified in \cref{subsec:specification}, the set of monomials $\mathcal M_\lambda$ can be constructed in polynomial time.
\end{lem}
\begin{proof}
  Since the length $\ell = \ell(w)$ of the longest Weyl group element is $\poly(\langle G \rangle)$, it suffices to show that, for each $r=1,\dots,\ell$, $\mathcal M_{\lambda,w_r}$ can be computed from $\mathcal M_{\lambda,w_{r-1}}$ in time $\poly(\langle G \rangle, \langle M \rangle)$.
  We first argue that $\mathcal N_{\lambda,w_{r-1}}$ as defined in \cref{eq:monomials for induction} can be constructed in polynomial time.
  For this, we note that both $w_{r-1}$ and $s_{\alpha_r} w(x)$ for any $x \in \mathcal M_{\lambda,r-1}$ are given by their reduced decompositions~\cite{lakshmibai}.
  But for any two Weyl group elements $w',w'' \in W$, given by their reduced decompositions, it can be decided in $\poly(\langle G \rangle)$ time whether $w' \leq w''$ (if $\mathfrak g$ is of classical type, this result can be deduced from~\cite[Theorems 5A, 5BC, 5D]{proctor}; the five exceptional Lie algebras can be treated separately in constant time).
  Since $\#\mathcal M_{\lambda,r-1} = \dim M_{\lambda,w_{r-1}} \leq \dim M_\lambda \leq \langle M \rangle$, it is now easy to see that $\mathcal N_{\lambda,w_{r-1}}$ can be constructed in time $\poly(\langle G \rangle, \langle M \rangle)$.
  After this, $\mathcal M_{\lambda,r}$ can be constructed via \cref{eq:monomial basis} likewise in $\poly(\langle G \rangle, \langle M \rangle)$.
\end{proof}

We now establish \cref{lem:multiset of weights in polynomial time,lem:representation matrices in polynomial time}:

\begin{proof}[Proof of \cref{lem:multiset of weights in polynomial time}]
  If $\mathfrak g$ is the compact real form of simple Lie algebra and $M$ is irreducible then this follows directly from \cref{lem:lakshmibai}: First compute $\mathcal M_\lambda$ and then add the weight $\beta(x)$ of $\pi(x) v_\lambda$ for all $x \in \mathcal M_\lambda$ into the multiset.
  If $\mathfrak g$ is one-dimensional abelian and $M$ irreducible then there is only a single weight, which we already know from the specification of $M$.
  If $\mathfrak g = \mathfrak g_1 \oplus \dots \oplus \mathfrak g_n$ is a direct sum of such Lie algebras and $M$ irreducible, then any irreducible representation is a tensor product of irreducible $\mathfrak g_i$-representations for $i=1,\dots,n$, and the multiset of weights can be identified with the Cartesian product of the multiset of weights of its constituents, which can be computed in polynomial time.
  Finally, if $M$ is reducible we apply the above procedure to each irreducible summand in its specification.
\end{proof}

\begin{proof}[Proof of \cref{lem:representation matrices in polynomial time}]
  We may likewise assume that $\mathfrak g_\CC$ is simple and $M$ is irreducible, i.e., $M = M_\lambda$ for some highest weight $\lambda$.
  It moreover suffices to show that the representation matrices of $e_\alpha$, $f_\alpha = e_{-\alpha}$, and $h_\alpha$ for the simple roots $\alpha \in PS(G)$ can be computed in polynomial time.
  Indeed, if $\beta \in P(G)$ is not simple then $e_\beta$ and $f_\beta$ can computed from the above by using $\poly(\langle G \rangle)$ many Lie brackets.
  We proceed case by case:

  \emph{Type A} ($\mathfrak g_\CC = \mathfrak{gl}_n(\CC)$):
  Let $\operatorname{GT}^*_\lambda$ denote the set of Gelfand-Tsetlin patterns, which can be computed in time $\poly(\langle G \rangle, \dim M_\lambda)$.
  Let $\operatorname{GT}_\lambda$ the corresponding Gelfand-Tsetlin basis; it is a basis of weight vectors.
  Using the explicit formulas given in~\cite{gelfand1,molev1} (cf.\ Theorem 2.3 in~\cite{molev1}), each matrix element of the corresponding representation matrices of $e_\alpha$, $f_\alpha$, and $h_\alpha$ can be computed in polynomial time given the pair of Gelfand-Tsetlin patterns indexing the entry.
  It follows that these representation matrices can be computed in time $\poly(\langle G \rangle, \dim M_\lambda)$.

  \emph{Types B, D} ($\mathfrak g_\CC = \mathfrak{so}_n(\CC)$):
  Using the corresponding Gelfand-Tsetlin basis~\cite{gelfand2}, the proof proceeds as for type A.

  \emph{Type C} ($\mathfrak g_\CC = \mathfrak{sp}_n(\CC)$):
  In this case, we can use Molev's basis~\cite{molev2} which has similar properties.

  \emph{Exceptional Types}:
  Since there are only finitely many exceptional Lie groups, they can all be embedded into some fixed $\GL_a(\CC)$.
  Correspondingly, $\mathfrak g_\CC$ can be embedded into $\mathfrak{gl}_a(\CC)$; let us fix such an embedding for each type.
  Now let $\mathcal O(\mathfrak{gl}_a(\CC))_l$ denote the degree-$l$ part of the coordinate ring of $\mathfrak{gl}_a(\CC)$.
  Write $\lambda = \sum_i \lambda_i \omega_i$, where the $\omega_i$ are the fundamental weights of $\mathfrak g$.
  Let $d := \sum_i d_i$; it is easy to see using Weyl's dimension formula that $d = O(\dim M_\lambda)$.
  It can be shown that $M_\lambda$ occurs as a subrepresentation of $M' := \mathcal O(\mathfrak{gl}_a(\CC))_l$ for some $l = O(d) = O(\dim M_\lambda)$ that can be computed explicitly from $\lambda$.
  The dimension of $M'$ is equal to the number of monomials of degree $l$ in $a^2$ variables, and therefore $O(l^{a^2}) = O(\poly(\dim M))$, since $a$ is a constant.

  We now compute the representation matrices of the generators of $\mathfrak{gl}_a(\CC)$ with respect to the basis of $M'$ consisting of the usual degree-$l$ monomials in $a^2$ variables, which we will denote by $\mathcal B'$.
  From this, we in turn obtain the representation matrices for the generators of $\mathfrak g_\CC$ with respect to the same basis by using the explicit embedding fixed above.
  After this, we can compute a basis of highest weight vectors for $G$ by finding those weight vectors that are annihilated by the action of the $e_\alpha$ for $\alpha \in PS(G)$.
  For each highest weight vector, we then compute the corresponding Lakshmibai monomial basis expressed in terms of the usual monomial basis by using \cref{lem:lakshmibai}.
  In this way we obtain a second basis $\mathcal B$ of $M'$.
  At last, we compute the inverse change of basis matrix and use it to express the representation matrices for the generators of $\mathfrak g_\CC$ with respect to the basis $\mathcal B$.
  By restricting these representation matrices to the Lakshmibai monomial basis $\mathcal B_\lambda \subseteq \mathcal B$ corresponding to a copy of $M = M_\lambda$ in $M'$, we finally obtain the desired representation matrices.
  All this can be done in time $O(\poly(\dim M')) = O(\poly(\dim M))$.
\end{proof}

\subsection{Inequalities for moment polytopes}

We now recall the description of the moment polytope $\Delta_G(M)$ from~\cite{vergnewalter2014inequalities}.
Let $\Phi(M) = \{ \varphi_1, \dots, \varphi_D \}$ denote the multiset of weights of the representation~$M$, where $D = \dim M$, and consider a corresponding decomposition of $M$ into one-dimensional weight spaces, $M = \bigoplus_{i=1}^D \CC \psi_i$, where $\psi_i$ is a weight vector of weight $\varphi_i$.
Recall that $N(G)$ denotes the set of negative roots of $G$ as defined at the beginning of \cref{sec:general}.
For any $H \in \Lambda_G$ and $z \in \ZZ$, we now define the following three sub(multi)sets:
\begin{align*}
  \Phi(H = z) &= \{ \varphi \in \Phi(M) : \varphi(H) = z \}, \\
  \Phi(H < z) &= \{ \omega \in \Phi(M) : \omega(H) < z \}, \\
  N(H < 0) &= \{ \alpha \in N(G) : \alpha(H) < 0 \}.
\end{align*}

For each root $\alpha$, we had defined a basis vector $e_\alpha \in i \mathfrak g$ in the corresponding root space.
Let $E_\alpha$ denote the linear operator given by the (complexified) representation of the Lie algebra of $G$ on $M$.
We will call $E_\alpha$ the \emph{root operator} corresponding to the root $\alpha$.

\begin{dfn}[\cite{vergnewalter2014inequalities}]
\label{dfn:ressayre general}
  A \emph{Ressayre element} is a pair $(H,z)$, where $H \in \Lambda_G$ and $z \in \ZZ$, such that the following conditions are satisfied:
  \begin{enumerate}
    \item \emph{Admissibility:} The points in $\Phi(H=z)$ span an affine hyperplane in $i \mathfrak t^*$.
    \item \emph{Trace condition:} $\# N(H<0) = \# \Phi(H<z)$.
    \item \emph{Determinant condition:}
    Consider a weight vector $\psi_j$ of weight $\varphi_j \in \Phi(H=z)$. Its image under a root operator $E_\alpha$ for $\alpha \in N(H<0)$ is necessarily a weight vector of some weight $\varphi_j + \alpha \in \Phi(H<z)$.
    Thus we can write
    \[ E_\alpha \psi_j = \sum_{i : \varphi_i = \varphi_j + \alpha} (D_{H,z,j})_{i,\alpha} \psi_i, \]
    whereby we obtain a matrix $D_{H,z,j}$ whose rows are indexed by integers $i$ with $\varphi_i \in \Phi(H<z)$ and whose columns are indexed by roots $\alpha \in N(H<0)$.
    Let $D_{H,z}$ denote the polynomial matrix
    \begin{equation}
      \label{eq:general det matrix}
      D_{H,z} = \!\!\!\! \sum_{j : \varphi_j \in \Phi(H=z)} D_{H,z,j} \, X_j
    \end{equation}
    in variables $X_j$.
    By the trace condition, $D_{H,z}$ is a square matrix, so that we can form the \emph{determinant polynomial} $d_{H,z} := \det D_{H,z}$, and the condition is that $d_{H,z}$ should be non-zero. 
  \end{enumerate}
\end{dfn}

As in \cref{subsec:kronecker ieqs}, we observe that the number of Ressayre elements is finite (up to overall rescaling), and -- assuming that $\Delta_G(M)$ is maximal-dimensional -- we have the following description of the moment polytope in terms of finitely many inequalities~\cite{vergnewalter2014inequalities}:
\begin{equation}
\label{eq:general ressayre}
  \Delta_G(M) = \{ r \in i \mathfrak t^*_+ : r(H) \geq z \text{ for all Ressayre elements $(H,z)$} \}
\end{equation}
We will call a facet of $\Delta_G(M)$ \emph{non-trivial} if it is \emph{not} a defining inequality of the Weyl chamber, i.e., if it is \emph{not} of the form $\lambda(h_\alpha) > 0$ for any of the simple coroots $h_\alpha$.
\Cref{eq:general ressayre} implies that any non-trivial facet is necessarily given by a Ressayre element.


\subsection{\textsc{MomentPolytope} is in coNP}
\label{subsec:general conp}

A problem instance for \textsc{MomentPolytope} is given by a quadruple $(G,M,\lambda,k)$ encoded as described in \cref{subsec:specification} above.
We now describe a polynomial-time algorithm that takes as input the problem instance together with a certificate that consists of a triple $(H,z,p)$, where $H \in \Lambda_G$, $z \in \ZZ$, and $p \in \ZZ^{\#\Phi(H=z)}$.
Here, $H$ is specified as an integer vector with respect to the bases fixed at the beginning of \cref{sec:general}.

The algorithm proceeds as follows:
We first check the conditions in \cref{dfn:ressayre general} to verify that $(H,z)$ is a Ressayre element for $\Delta_G(M)$:
\begin{enumerate}
  \item Admissibility:
    There are $\#\Phi(M) = \dim M = O(\langle M \rangle)$ weights, each living in a space of dimension $\dim T = O(\langle G \rangle)$.
    According to \cref{lem:multiset of weights in polynomial time}, we can compute the multiset of weights $\Phi(M)$ in polynomial time.
    For each weight $\varphi\in\Phi(M)$, we can determine if $\varphi\in\Phi(H=z)$ by verifying that $\varphi(H) = z$, which amounts to evaluating an inner product in $\ZZ^{\dim T}$.
    Thus we can in polynomial time determine $\Phi(H=z)$ and compute the rank of the polynomial-size matrix with columns $\begin{psmallmatrix}\varphi \\ -1\end{psmallmatrix}$ for $\varphi\in\Phi(H=z)$.
    The element $(H,z)$ is admissible if and only if the rank is equal to $\dim T$.
  \item Trace condition:
    As there are no more than $O(\langle G \rangle^2)$ negative roots and $\dim M = O(\langle M \rangle)$ weights, each of which lives in a space of dimension $\dim T = O(\langle G \rangle)$ and can be constructed efficiently (\cref{lem:multiset of weights in polynomial time}), both cardinalities can be computed and compared in polynomial time.
  \item Determinant:
    We construct the matrix $D_{H,z}(p)$ defined in \cref{eq:general det matrix} for $X = p$.
    The matrix is of polynomial size and can be constructed efficiently (\cref{lem:representation matrices in polynomial time}).
    We can therefore compute its determinant $d_{H,z}(p)$ exactly in polynomial time.
    We accept if and only if $d_{H,z}(p) \neq 0$.
\end{enumerate}
At this point we are sure that $(H,z)$ defines a non-trivial facet of the moment polytope (the trivial inequalities are automatically satisfied since $\lambda$ is a highest weight and therefore an element of the positive Weyl chamber $i \mathfrak t^*_+$).
In the last step of the algorithm, we verify that this facet indeed separates $\lambda/k$ from the polytope by checking that $H \cdot \lambda < k z$.
It is clear that the algorithm will accept only if $\lambda/k \not\in \Delta_G(M)$.

We will now show that, conversely, if $\lambda/k \not\in \Delta_G(M)$ then there always exists a polynomial-sized certificate $(H,z,p)$ such that the algorithm accepts.
For this, we derive the following estimates:

\begin{lem}
\label{lem:coeff bound}
  Let $\varphi \in \Phi(M)$. Then $\lVert \varphi \rVert_\infty \leq 2^{\langle M \rangle}$, where we think of $\varphi$ as an integer vector with respect to the basis of $\Lambda^*_G$ fixed at the beginning of \cref{sec:general}.
\end{lem}
\begin{proof}
  We may assume that $\mathfrak g$ is simple and that $M$ is an irreducible representation of highest weight~$\nu$.
  In this case, $\varphi$ is specified with respect to the basis of fundamental weights, i.e., the coefficients of $\varphi$ are given by $\varphi(h_\alpha)$ where $\alpha$ ranges over the simple coroots.
  To bound $\varphi(h_\alpha)$, we use the classical fact that the convex hull of the weights $\Phi(M)$ is equal to the convex hull of the Weyl group orbit $W \cdot \nu$ of the highest weight (e.g., \cite[Theorem 7.41]{hall2015lie}).
  Therefore,
  \[
         \lvert \varphi(h_\alpha) \rvert
    \leq \max_{w \in W} \lvert \nu(w \cdot h_\alpha) \rvert
    \leq \max_{h_\beta} \lvert \nu(h_\beta) \rvert = \lVert \nu \rVert_\infty,
  \]
  where the last maximization is over all simple coroots $h_\beta$.
  This shows that $\lVert \varphi \rVert_\infty \leq  \lVert \nu \rVert_\infty$.
  As $M$ is specified in terms of the coefficients of the highest weight $\nu$ given \emph{in binary}, this shows that
  \[
  	\lVert \nu \rVert_\infty \leq 2^{\langle\nu\rangle} \leq 2^{\langle M \rangle}. 
  \]
\end{proof}

\begin{lem}
\label{lem:general siegel}
  Any non-trivial facet of $\Delta_G(M)$ can be described by a Ressayre element $(H,z)$ with
  \begin{equation}
  \label{eq:general siegel}
    \max~\{ \lVert H \rVert_\infty, \lvert z \rvert \} \leq 2^{\langle G \rangle (\log (\langle G \rangle + 1) + \langle M\rangle)},
  \end{equation}
  where we think of $H$ as an integer with respect to the basis of $\Lambda_G$ fixed at the beginning of \cref{sec:general}.
\end{lem}
\begin{proof}
  The admissibility condition in \cref{dfn:ressayre general} implies that any Ressayre element $(H,z)$ is the normal vector of an affine hyperplane in $i \mathfrak t^*$ spanned by some affinely independent set of weights $\varphi_1, \dots, \varphi_R \in \Phi(H=z)$, where $R = \dim T$.
  Therefore, $(H,z) \in \ZZ^{R+1}$ is an integral solution to the following linear system of equations:
  \begin{align*}
    &\varphi_1(H) - z = 0, \;\dots,\; \varphi_R(H) - z = 0.
  \end{align*}
  Note that we have $R$ equations for $R + 1$ unknowns, and the absolute value of the coefficients is at most $2^{\langle M\rangle}$ by \cref{lem:coeff bound} above.
  Therefore, Siegel's lemma~\cite{hindry2000diophantine} ensures that there exists an integral solution with
  \[   \max~\{ \lVert H \rVert_\infty, \lvert z \rvert \}
  \leq ((R+1)2^{\langle M\rangle})^{R/((R+1)-R)}
    = ((R+1)2^{\langle M\rangle})^R
  \leq 2^{R (\log (R+1) + \langle M\rangle)}. \]
  Since $R \leq \langle G \rangle$ we obtain the claim of the lemma.
\end{proof}

The upshot of \cref{lem:general siegel} is the following:
If $\lambda/k \not\in \Delta_G(M)$ then there exists a non-trivial facet separating it from the moment polytope.
\Cref{lem:general siegel} tells us that any such facet can be encoded by some Ressayre element $(H,z)$ that can be specified using no more than $O(\langle G \rangle^2 (\log \langle G \rangle + \langle M \rangle))$ bits.
Indeed, $(H,z)$ consists of $\dim T + 1 = O(\langle G \rangle)$ coefficients, each of which requires $O(\langle G \rangle (\log \langle G \rangle + \langle M \rangle))$ bits.

Now consider the determinant polynomial $d_{H,z}$, which is a multivariate polynomial of degree $\#\Phi(H<z) 
 \leq \dim M \leq \langle M \rangle$ in $\#\Phi(H=z) \leq \dim M \leq \langle M \rangle$ variables.
As before, we can use the Schwartz-Zippel lemma to deduce the existence of a point $p \in \{0,\dots,\langle M \rangle\}^{\#\Phi(H=z)}$ such that $d_{H,z}(p)\neq0$.
Note that $p$ can be encoded using no more than $O(\langle M \rangle \log \langle M \rangle)$ bits.

We have thus obtained a polynomial-sized certificate $(H,z,p)$ that will be accepted by the algorithm.
We conclude that the problem \textsc{MomentPolytope} is in $\coNP$.

\subsection{\textsc{MomentPolytope} is in NP}
\label{subsec:general np}

We now show that \textsc{MomentPolytope} is also in $\NP$.
As in the case of the Kronecker polytope, we will use the geometric description from~\cite{ness1984stratification}.
For any non-zero vector $\psi \in M$, define its \emph{image under the moment map} $\mu(\psi) \in i \mathfrak g^*$ by the following formula:
\begin{equation}
\label{eq:moment map}
  \forall x \in i\mathfrak g\!: \qquad \mu(\psi)(x) = \frac {\braket{\psi | \pi(x) | \psi}} {\braket{\psi | \psi}},
\end{equation}
where $\pi$ denotes the (complexified) representation of the Lie algebra of $G$ on $M$.
Let $r(\psi)$ denote the unique point of intersection of the coadjoint $G$-orbit through $\mu(\psi)$ with the positive Weyl chamber~$i \mathfrak t^*_+$.
Then we have the following characterization of the moment polytope~\cite{ness1984stratification}:
\begin{equation}
\label{eq:general polytope geometric}
  \Delta_M(G) = \{ r(\psi) : 0 \neq \psi \in M \}.
\end{equation}

We now describe a polynomial-time algorithm that takes as input the problem instance $(G,M,\lambda,k)$ together with a certificate that consists of a non-zero vector $\tilde\psi \in \QQ[i]^{\#\Phi(M)}$.
We first compute $\mu(\tilde\psi)$ as a rational vector with respect to the basis of $i \mathfrak g^*$ fixed at the beginning of \cref{sec:general}; this can be done in polynomial time by using \cref{lem:representation matrices in polynomial time}.
We then compute $d^2 := \lVert \mu(\tilde\psi) - \lambda/k \rVert_{i \mathfrak g^*}^2$; this can again be done in polynomial time.
Finally, we accept if and only if
\begin{equation}
\label{eq:general accept}
  d = \lVert \mu(\tilde\psi) - \lambda/k \rVert_{i \mathfrak g^*}
  \leq \frac12 \frac 1 k \frac 1 {4^{\langle G \rangle (\log \langle G \rangle + \langle M \rangle)}}
\end{equation}
Here, we think of the highest weight $\lambda \in i \mathfrak t^*$ as an element in $i \mathfrak g^*$ by extending by zero on the root spaces.

Let us first analyze the case where the point $\lambda/k \in \Delta_G(M)$.
According to \cref{eq:general polytope geometric}, there exists a unit vector $\psi \in M$ such that $r(\psi) = \lambda/k$.
As the moment map $\psi \mapsto \mu(\psi)$ is $G$-equivariant, we may in fact arrange for $\mu(\psi) = \lambda/k$ to hold.
Now let $\tilde\psi \in \QQ[i]^{\#\Phi(M)}$ be a truncation of $\psi$ to $b$ bits in the real and imaginary part of each of its $\#\Phi(M) = O(\langle M \rangle)$ components with respect to the weight basis from \cref{lem:representation matrices in polynomial time}.
The following lemma bounds the resulting error on the image under the moment map as a function of the precision $b$:

\begin{lem}
\label{lem:general truncation}
  We have that
  \[ \lVert \mu(\tilde\psi) - \mu(\psi) \rVert_{i \mathfrak g^*} = O(\langle G \rangle^{1/2} \langle M \rangle^{1/4} 2^{\langle M \rangle} 2^{-b/2}). \]
\end{lem}
\begin{proof}
  For any $x \in i \mathfrak g$, we have from \cref{eq:moment map,lem:hilbert space vs trace norm} that
  \[
    \lvert \mu(\psi)(x) - \mu(\tilde\psi)(x) \rvert
    \leq \lvert \tr (P_{\tilde\psi} - P_\psi) \pi(x) \rvert
    \leq \lVert P_{\tilde\psi} - P_\psi \rVert_1 \; \lVert \pi(x) \rVert
    \leq 5 \langle M \rangle^{1/4} 2^{-b/2} \; \lVert \pi(x) \rVert,
  \]
  where $\lVert \pi(x) \rVert$ denotes the operator norm of $\pi(x)$.
  Therefore,
  \begin{equation}
  \label{eq:general truncation mid}
    \lVert \mu(\psi) - \mu(\tilde\psi) \rVert_{i \mathfrak g^*}
    = \max_{\lVert x \rVert_{i \mathfrak g} = 1} \lvert \mu(\psi)(x) - \mu(\tilde\psi)(x) \rvert
    \leq 5 \langle M \rangle^{1/4} 2^{-b/2} \max_{\lVert x \rVert_{i \mathfrak g} = 1} \lVert \pi(x) \rVert.
  \end{equation}
  To compute the right-hand side maximum, we recall that $\lVert-\rVert_{i \mathfrak g}$ is invariant under the adjoint action, the operator norm is (in particular) invariant under conjugation by unitaries, and $\pi$ is correspondingly equivariant, so that we may restrict the maximization to $x \in i \mathfrak t$.
  Then $\pi(x)$ acts as a multiplication operator in the weight basis, multiplying weight vectors of weight $\varphi$ by $\varphi(x)$.
  We may assume without loss of generality that $M$ is an irreducible representation with some highest weight $\nu$, so that
  \[
      \lVert \pi(x) \rVert
    = \max_{\varphi \in \Phi(M)} \lvert \varphi(x) \rvert
    = \max_{w \in W} \lvert \nu(w \cdot x) \rvert,
  \]
  where we have again used that the convex hull of weights is equal to the convex hull of the Weyl group orbit of the highest weight (cf.\ the proof of \cref{lem:coeff bound}).
  It follows that
  \begin{equation}
  \label{eq:general trunction end}
      \max_{\lVert x \rVert_{i \mathfrak g} = 1} \lVert \pi(x) \rVert
    = \max_{x \in i \mathfrak t : \lVert x \rVert_{i \mathfrak g} = 1} \lvert \nu(x) \rvert
    = \lVert \nu \rVert_{i \mathfrak g^*}
    = O(\langle G \rangle^{1/2} 2^{\langle M \rangle}),
  \end{equation}
  where the last estimate follows from observing that the highest weight $\nu$ is specified in terms of its coefficients (in binary) with respect to the basis vectors of the weight lattice, which have norm $\Theta(1)$.
  The asserted bound follows from plugging \cref{eq:general trunction end} back into \cref{eq:general truncation mid}.
\end{proof}

As a direct consequence of \cref{lem:general truncation}, the truncation to $b$ bits leads to an error of at most
\[
    \lVert \mu(\tilde\psi) - \lambda/k \rVert_{i \mathfrak g^*}
  = \lVert \mu(\tilde\psi) - \mu(\psi) \rVert_{i \mathfrak g^*}
  = O(\langle G \rangle^{1/2} \langle M \rangle^{1/4} 2^{\langle M \rangle} 2^{-b/2}).
\]
Comparing with \cref{eq:general accept}, we find that it suffices to choose $b = O(\log k + \langle G \rangle (\log \langle G \rangle + \langle M \rangle)$ bits of precision to produce a certificate that the algorithm accepts.
This is polynomial in the size of the problem instance.

%

\medskip

Conversely, let us assume that in fact $\lambda/k \not\in \Delta_G(M)$.
We will use the following lemma:

\begin{lem}
\label{lem:distance general}
  Let $\lambda \in \Lambda^*_G$ and $k > 0$. Then, if $\lambda/k \not\in \Delta_G(M)$, it has $\lVert-\rVert_{i \mathfrak g^*}$-distance at least
  \[ \frac 1 k \frac 1 {4^{\langle G \rangle (\log \langle G \rangle + \langle M \rangle)}} \]
  to the moment polytope.
\end{lem}
\begin{proof}
  Consider a non-trivial facet that separates $\lambda/k$ from $\Kron(m)$.
  According to \cref{lem:general siegel}, any such facet can be described by some Ressayre element $(H,z)$ satisfying \cref{eq:general siegel}.
  We can therefore lower-bound the distance of $\lambda/k$ to $\Delta_G(m)$ by
  \[
    \frac {\lvert \lambda(H)/k - z \rvert}{\lVert H \rVert_{i \mathfrak g}}
    = \frac 1 k \frac {\lvert \lambda(H) - k z \rvert} {\lVert H \rVert_{i \mathfrak g}}
    \geq \frac 1 k \frac 1 {\lVert H \rVert_{i \mathfrak g}}.
  \]
  Recall that we may also think of $H$ as an integer vector with respect to the basis vectors fixed at the beginning of \cref{sec:general}.
  As the latter have norm $\Theta(1)$, we obtain that $\lVert H \rVert_{i \mathfrak g} \leq \sqrt{\langle G \rangle} \lVert H \rVert_\infty$.
  Together with the upper bound \cref{eq:general siegel}, we find that
  \[
    \frac 1 k \frac 1 {\lVert H \rVert_{i \mathfrak g}}
    \geq \frac 1 k \frac 1 {\langle G \rangle^{1/2} \lVert H \rVert_\infty}
    \geq \frac 1 k \frac 1 {\langle G \rangle^{1/2} 2^{\langle G \rangle (\log(\langle G \rangle + 1) + \langle M \rangle)}}
    \geq \frac 1 k \frac 1 {4^{\langle G \rangle (\log \langle G \rangle + \langle M \rangle)}}.
  \]
\end{proof}

It follows from \cref{lem:distance general,eq:general polytope geometric} that the distance of $r(\tilde\psi)$ for any vector $0 \neq \tilde\psi \in M$ to $\lambda/k$ is never smaller than $1/k \cdot 1/4^{\langle G \rangle (\log \langle G \rangle + \langle M \rangle)}$.
This implies that
\[
     \frac 1 k \frac 1 {4^{\langle G \rangle (\log \langle G \rangle + \langle M \rangle)}}
\leq \lVert r(\tilde\psi) - \lambda/k \rVert_{i \mathfrak g^*}
\leq \lVert \mu(\tilde\psi) -\lambda/k \rVert_{i \mathfrak g^*},
\]
where the second inequality is~\cite[Lemma 4.10]{walter2014thesis}.
In view of the acceptance condition in \cref{eq:general accept}, we find that our algorithm will therefore never accept if $\lambda/k \not\in \Kron(m)$.
We conclude that the problem \textsc{KronPolytope} is in $\NP$.
\Cref{subsec:general conp,subsec:general np} together establish \cref{thm:general}.

\section*{Acknowledgments}

We acknowledge pleasant discussions with Christian Ikenmeyer.
We thank the Simons Institute for the Theory of Computing, the American Institute of Mathematics, and the organizers of the workshop on ``Combinatorics and complexity of Kronecker coefficients'', where this work has been initiated.

\bibliographystyle{siamplain}
\bibliography{momentcomplex}

\end{document}